\documentclass{new_tlp}

\usepackage{times}
\usepackage{soul}
\usepackage{url}
\usepackage[hidelinks]{hyperref}
\usepackage[utf8]{inputenc}
\usepackage[small]{caption}
\usepackage{graphicx}
\usepackage{amsmath}
\usepackage{booktabs}
\usepackage{algorithm}
\usepackage{algorithmic}
\title{On Uniform Equivalence of Epistemic Logic Programs}

\author[W.\ Faber, M.\ Morak, and S.\ Woltran]{
  WOLFGANG FABER$^1$, MICHAEL MORAK$^1$, and STEFAN WOLTRAN$^2$\\
  $^1$University of Klagenfurt, Austria\\
  $^2$TU Wien, Vienna, Austria\\
  \email{wolfgang.faber@aau.at, michael.morak@aau.at, woltran@dbai.tuwien.ac.at}
}

\jdate{March 2003}
\pubyear{2003}
\pagerange{\pageref{firstpage}--\pageref{lastpage}}
\doi{S1471068401001193}

\usepackage[OT1]{fontenc}
\usepackage{amssymb}
\usepackage{wasysym}
\usepackage{multicol}

\newcommand{\nop}[1]{}

\newenvironment{changemargin}[2]{%
\list{}{\rightmargin#2\leftmargin#1
\parsep=0pt\topsep=0pt\partopsep=0pt}
\item[]}
{\endlist}

\newenvironment{indented}{\begin{changemargin}{1cm}{0cm}}{\end{changemargin}}

\newtheorem{theorem}{Theorem}

\newtheorem{proposition}[theorem]{Proposition}

\newtheorem{definition}[theorem]{Definition}
\newtheorem{example}[theorem]{Example}

\let\phi\varphi
\let\epsilon\varepsilon

\renewcommand{\models}{\vDash}

\newcommand{\calA}{\mathcal{A}}

\newcommand{\calC}{\mathcal{C}}

\newcommand{\calE}{\mathcal{E}}

\newcommand{\calI}{\mathcal{I}}

\newcommand{\calM}{\mathcal{M}}

\newcommand{\calR}{\mathcal{R}}
\newcommand{\calS}{\mathcal{S}}

\newcommand{\calU}{\mathcal{U}}

\newcommand{\NP}{\ensuremath{\textsc{NP}}}
\newcommand{\co}{\ensuremath{\textsc{co}}}

\newcommand{\SIGMA}[2]{\ensuremath{\Sigma_{\mathit{#1}}^{\mathit{#2}}}}

\newcommand{\DP}[2]{\ensuremath{D_{\mathit{#1}}^{\mathit{#2}}}}
\newcommand{\PI}[2]{\ensuremath{\Pi_{\mathit{#1}}^{\mathit{#2}}}}

\newcommand{\variables}[1]{{\mathbf{#1}}}

\newcommand{\mods}[1]{\mathit{mods}(#1)}
\newcommand{\answersets}[1]{\mathit{AS}(#1)}
\newcommand{\semods}[1]{\calS\calE(#1)}
\newcommand{\uemods}[1]{\calU\calE(#1)}
\newcommand{\cwvs}[1]{\mathit{cwv}(#1)}

\newcommand{\varsX}{\variables{X}}
\newcommand{\varsY}{\variables{Y}}
\newcommand{\varsZ}{\variables{Z}}

\newcommand{\relation}[1]{{\mathit{#1}}}

\newcommand{\eneg}{\mathbf{not}\,}

\newcommand{\sneg}{\mathbf{\sim}\,}

\newcommand{\body}[1]{{\mathit{B}(#1)}}

\newcommand{\pbody}[1]{{\mathit{B}^+(#1)}}
\newcommand{\head}[1]{{\mathit{H}(#1)}}

\begin{document}

\maketitle

\begin{abstract}
  Epistemic Logic Programs (ELPs) extend Answer Set Programming (ASP) with
epistemic negation and have received renewed interest in recent years. This led
to the development of new research and efficient solving systems for ELPs. In
practice, ELPs are often written in a modular way, where each module interacts
with other modules by accepting sets of facts as input, and passing on sets of
facts as output. An interesting question then presents itself: under which
conditions can such a module be replaced by another one without changing the
outcome, for any set of input facts? This problem is known as uniform
equivalence, and has been studied extensively for ASP. For ELPs, however, such
an investigation is, as of yet, missing. In this paper, we therefore propose a
characterization of uniform equivalence that can be directly applied to the
language of state-of-the-art ELP solvers. We also investigate the computational
complexity of deciding uniform equivalence for two ELPs, and show that it is on
the third level of the polynomial hierarchy.

\end{abstract}

\section{Introduction}\label{sec:introduction}

Epistemic Logic Programs (ELPs) \cite{aaai:Gelfond91,logcom:KahlWBGZ15,%
ai:ShenE16} add epistemic operators to the language of Answer Set Programming
(ASP) \cite{book:GebserKKS12,cacm:BrewkaET11,ki:SchaubW18}, a generic, fully
declarative logic programming language that allows for encoding problems such
that the resulting answers (called \emph{answer sets}) directly correspond to
solutions of the encoded problem. In ASP, negation is usually interpreted
according to the stable model semantics \cite{iclp:GelfondL88}, that is, as
negation-as-failure or default negation. Intuitively, a default negated atom
$\neg a$ is true if there is no justification for $a$ being true in the same
answer set. Hence, default negation is a ``local'' operator defined relative to
the answer set under consideration. ELPs (as defined in \cite{ai:ShenE16})
extend ASP with the epistemic negation operator $\eneg$. An epistemically
negated atom $\eneg a$ intuitively means that atom $a$ cannot be \emph{proven}
to be true, in the sense that it is not true in every answer set. Epistemic
negation is thus defined with respect to a set of answer sets, referred to as a
\emph{world view}. Deciding whether such a world view exists is
\SIGMA{P}{3}-complete in general \cite{ai:ShenE16}, whereas deciding answer set existence
for ASP can be done in \SIGMA{P}{2} \cite{amai:EiterG95}, one level lower on the
polynomial hierarchy.

Michael Gelfond \shortcite{aaai:Gelfond91,amai:Gelfond94} recognized epistemic
negation to be a useful construct for ASP early on and proceeded to introduce
the modal operators $\mathbf{K}$ (``known'' or ``provably true'') and
$\mathbf{M}$ (``possible'' or ``not provably false'') to add this feature to the
language. $\mathbf{K}a$ and $\mathbf{M}a$ correspond to $\neg \eneg a$ and
$\eneg \neg a$, respectively. Renewed interest in recent years has revealed
several flaws in the original semantics, and various new approaches (cf.\ e.g.\
\cite{lpnmr:Gelfond11,birthday:Truszczynski11,diss:Kahl14,ijcai:CerroHS15,%
ai:ShenE16}) were proposed. Also, several efficient and practical ELP solving
systems have been and continue to be developed \cite{logcom:KahlWBGZ15,%
ijcai:SonLKL17,ijcai:BichlerMW18}.

An interesting question in the context of ELPs is when two programs are
equivalent. Standard (sometimes also called ordinary or classical) equivalence
is simply defined as two programs having the same world views. However, as for
ASP but unlike for classical logic, this notion does not capture replaceability
of ELPs. In order to be able to capture when a program is replaceable by another
one, the context in which the replacement is done has to be taken into account.
The notion of strong equivalence captures replaceability in any context. Strong
equivalence is a well-studied topic in ASP (there the context can be any set of
ASP rules) with several useful applications \cite{tocl:LifschitzPV01,%
tplp:Turner03,iclp:CabalarPV07,jair:LinC07,jancl:EiterFPTW13}. An analogous
notion has also been studied more recently for ELPs (in that case, the context
can be any set of ELP rules) \cite{lpnmr:WangZ05,ijcai:CerroHS15,%
aaai:FaberMW19}.

For some applications, other notions of equivalence between these two extremes
(no context for standard equivalence, any context for strong equivalence) are
desirable. The most prominent of these intermediate equivalences is uniform
equivalence, where the context is restricted to facts. This notion had been
proposed for Datalog originally \cite{Sagiv88,Maher88} and has been studied for
ASP quite extensively as well \cite{iclp:EiterF03,lpnmr:EiterFTW04,%
aaai:EiterFTW05}. The notion of uniform equivalence is of course most suitable
when considering a full program that is applied to various scenarios, which are
represented as factual knowledge, or a program that serves as a fixed knowledge
base for an agent, which then uses it together with percepts represented as
facts. The concept is also useful for modular ASP programs
\cite{iclp:LifschitzT94,iclp:Oikarinen07,jair:JanhunenOTW09}, where sub-programs
interact with each other by accepting a set of input facts and returning a set
of output facts.

While uniform equivalence has been widely studied for ASP, such an investigation
is, to the best of the authors' knowledge, still lacking for ELPs.

\begin{example}\label{ex:running1}
  From \cite{aaai:FaberMW19}, we take the example of formulating the well-known
  Closed World Assumption via ELP rules. Two formulations of the CWA have been
  proposed in this context. In \cite{aaai:Gelfond91}, a rule for CWA is proposed
  that, in the language of ELPs, can be formulated as \[ p' \gets \neg \eneg
  \neg p. \] Intuitively, this says that $p'$ (meaning the negation of $p$)
  shall be true if there is no possible world where $p$ is true. In
  \cite{ai:ShenE16}, a different rule is proposed: \[ p' \gets \eneg p. \] Here,
  intuitively, $p'$ shall be true if there is a possible world where $p$ is
  false. While the two formulations are equivalent in that they share the same
  world views, in \cite{aaai:FaberMW19}, it was shown that they are, however,
  not strongly equivalent. However, this does not tell us anything about the
  uniform equivalence of these two rules.
\end{example}

In order to analyze cases like the one presented in Example~\ref{ex:running1}
above, it is the aim of this paper to study uniform equivalence for ELPs, that
is, the question of whether, given two ELPs $\Pi_1$ and $\Pi_2$, for any set of
facts $D$, the combined programs $\Pi_1 \cup D$ and $\Pi_2 \cup D$ have the same
world views. According to \cite{aaai:FaberMW19}, two versions of (ordinary)
equivalence between ELPs can be defined: one where all candidate world views are
equal, and one where only the world views (that is, candidate world views that
minimize the number of assumptions) are equal. In \cite{aaai:FaberMW19}, two
versions of strong equivalence, relative to these ordinary equivalence notions,
are defined, and then subsequently shown to coincide. We will follow the same
approach. Interestingly, we will see that for uniform equivalence, the two
versions do not coincide.

\paragraph{Contributions.} The main contributions of this paper are the
following:
\begin{itemize}
  \item We formally define two versions of uniform equivalence for ELPs (based
    on the input language of today's ELP solvers) that appropriately extend
    uniform equivalence for ASP, based on existing notions of equivalence for
    ELPs, as used in \cite{aaai:FaberMW19}.

  \item We provide an analysis of the two different notions of uniform
    equivalence for ELPs and characterize their relationship. Furthermore, a
    model-theoretic characterization is offered, based on a so-called
    UE-function, in the same vein as the SE-function was introduced in
    \cite{aaai:FaberMW19} to characterize strong equivalence.
  \item We then show that testing uniform equivalence of two ELPs is
    \PI{3}{P}-complete, that is, the complexity of this test jumps up one level
    on the polynomial hierarchy compared to ASP, and hence is much harder than
    testing strong equivalence, which for ELPs is only \co\NP-complete
    \cite{aaai:FaberMW19}.
\end{itemize}

\paragraph{Organization.} The remainder of the paper is structured as follows.
Section~\ref{sec:preliminaries} gives an overview of the relevant definitions
needed in the main sections of the paper, including the language of ASP, ELPs,
and the notions of strong and uniform equivalence for the former.
Section~\ref{sec:uniformequivalence} defines two different notions of uniform
equivalence for ELPs, shows that, in contrast to strong equivalence, these
notions do not coincide, and finally offers a model-theoretic characterization
of these notions of uniform equivalence, called the UE-function. Following this
characterization, we investigate the computational complexity of deciding
uniform equivalence in Section~\ref{sec:complexity}. We then offer some
concluding remarks in Section~\ref{sec:conclusions}.

\section{Preliminaries}\label{sec:preliminaries}

\paragraph{Answer Set Programming (ASP).} A \emph{ground logic program} with
nested negation (also called answer set program, ASP program, or, simply, logic
program) is a pair $\Pi = (\calA, \calR)$, where $\calA$ is a set of
propositional (i.e.\ ground) atoms and $\calR$ is a finite set of rules of the
form
\begin{equation}\label{eq:rule}
  a_1\vee \cdots \vee a_l \leftarrow a_{l+1}, \ldots, a_m, \neg \ell_1, \ldots,
  \neg \ell_n;
\end{equation}
where the comma symbol stands for conjunction, $0 \leq l \leq m$, $0 \leq n$,
$a_i \in \calA$ for all $1 \leq i \leq m$, and each $\ell_i$ is a
\emph{literal}, that is, either an atom $a$ or its (default) negation $\neg a$
for any atom $a \in \calA$. Note that, therefore, doubly negated atoms may
occur. We will sometimes refer to such rules as \emph{standard rules}.  Each
rule $r \in \calR$ of form~(\ref{eq:rule}) consists of a \emph{head} $\head{r} =
\{ a_1,\ldots,a_l \}$ and a \emph{body} $\body{r} = \{a_{l+1},\ldots,a_m, \neg
\ell_1, \ldots, \neg \ell_n \}$. We denote the \emph{positive} body by
$\pbody{r} = \{ a_{l+1}, \ldots, a_m \}$. A rule where $l = 1$, $m = l$, and $n
= 0$ is called a \emph{fact}.

An \emph{interpretation} $I$ (over $\calA$) is a set of atoms, that is, $I \subseteq \calA$.  A
literal $\ell$ is true in an interpretation $I \subseteq \calA$, denoted $I
\models \ell$, if $a \in I$ and $\ell = a$, or if $a \not\in I$ and $\ell = \neg
a$; otherwise $\ell$ is false in $I$, denoted $I \not\models \ell$. Finally, for
some literal $\ell$, we define that $I \models \neg \ell$ if $I \not\models
\ell$. This notation naturally extends to sets of literals. An interpretation
$M$ is called a \emph{model} of $r$, denoted $M \models r$, if, whenever $M
\models \body{r}$, it holds that $M \models \head{r}$. We denote the set of
models of $r$ by $\mods{r}$; the models of a logic program $\Pi= (\calA,\calR)$
are given by $\mods{\Pi} = \bigcap_{r \in \calR} \mods{r}$. We also write
$I\models r$ (resp.\ $I\models \Pi$) if $I\in\mods{r}$ (resp.\
$I\in\mods{\Pi}$).

The GL-reduct $\Pi^I$ of a 
logic program $\Pi = (\calA, \calR)$ with
respect to an interpretation $I$ is the program $
(\calA, \calR^I)$,
where $\calR^I = \{ \head{r} \leftarrow \pbody{r} \mid r \in \calR, \forall \neg
\ell \in \body{r} : I \models \neg \ell \}$.

\begin{definition}\label{def:answerset}
  \cite{iclp:GelfondL88,ngc:GelfondL91,amai:LifschitzTT99} $M \subseteq \calA$
  is an \emph{answer set} of a logic program $\Pi$ if (1) $M \in \mods{\Pi}$ and
  (2) there is no subset $M' \subset M$ such that $M' \in \mods{\Pi^M}$.
\end{definition}

The set of answer sets of a logic program $\Pi$ is denoted by $\answersets{\Pi}$.
The \emph{consistency problem} of ASP, that is, to decide whether for a given
logic program $\Pi$ it holds that $\answersets{\Pi}\neq\emptyset$, is
\SIGMA{P}{2}-complete~\cite{amai:EiterG95}, and remains so also in the case
where doubly negated atoms are allowed in rule bodies~\cite{tplp:PearceTW09}.

\paragraph{Strong and Uniform Equivalence for Logic Programs.} Two logic
programs $\Pi_1 = (\calA, \calR_1)$ and $\Pi_2 = (\calA, \calR_2)$ are
\emph{equivalent} iff they have the same set of answer sets, that is,
$\answersets{\Pi_1} = \answersets{\Pi_2}$. The two logic programs are
\emph{strongly equivalent} iff for any third logic program $\Pi = (\calA,
\calR)$ it holds that the combined logic program $\Pi_1 \cup \Pi = (\calA,
\calR_1 \cup \calR)$ is equivalent to the combined logic program $\Pi_2 \cup \Pi
= (\calA, \calR_2 \cup \calR)$. They are \emph{uniformly equivalent} iff they
are strongly equivalent for any third program $\Pi$ consisting only of facts.
An \emph{SE-model} \cite{tplp:Turner03} of a logic program $\Pi = (\calA,
\calR)$ is a tuple of interpretations $(X, Y)$, where $X \subseteq
Y\subseteq\calA$, $Y \models \Pi$, and $X \models \Pi^Y$. The set of SE-models
of a logic program $\Pi$ is denoted $\semods{\Pi}$. Note that for every model
$Y$ of $\Pi$, $(Y, Y)$ is an SE-model of $\Pi$, since $Y \models \Pi$ implies
that $Y \models \Pi^Y$. An SE-model $(X, Y)$ of $\Pi$ is a \emph{UE-model} of
$\Pi$ \cite{iclp:EiterF03} iff either $X = Y$, or $X \subset Y$ and there is no
other SE-model $(X', Y) \in \semods{\Pi}$ such that $X \subset X' \subset Y$.
The set of UE-models of $\Pi$ is denoted $\uemods{\Pi}$. Hence, UE-models are
precisely those SE-models, where the $X$ component is either $Y$, or
subset-maximal w.r.t.\ the other SE-models.

Two logic programs (over the same set of atoms) are uniformly equivalent iff
they have the same UE-models and checking uniform equivalence is
\PI{2}{P}-complete in general \cite{iclp:EiterF03}.

\paragraph{Epistemic Logic Programs.} An \emph{epistemic literal} is a formula
$\eneg \ell$, where $\ell$ is a literal and $\eneg$ is the epistemic negation
operator. A \emph{ground epistemic logic program (ELP)} is a triple $\Pi =
(\calA, \calE, \calR)$, where $\calA$ is a set of propositional atoms, $\calE$
is a set of epistemic literals over the atoms $\calA$, and $\calR$ is a finite
set of \emph{ELP rules}, which are
\begin{equation*}
   a_1\vee \cdots \vee a_k \leftarrow \ell_1, \ldots, \ell_m, \xi_1, \ldots,
   \xi_j, \neg \xi_{j + 1}, \ldots, \neg \xi_{n},
\end{equation*}
where each $a_i\in\calA$ is an atom, each $\ell_i$ is a literal, and each $\xi_i
\in \calE$ is an epistemic literal.  Note that usually $\calE$ is defined
implicitly to be the set of all epistemic literals appearing in the rules
$\calR$; however, making the domain of epistemic literals explicit will prove
useful for our purposes.
%
%
The \emph{union} of two ELPs $\Pi_1 = (\calA_1, \calE_1, \calR_1)$ and $\Pi_2 =
(\calA_2, \calE_2, \calR_2)$ is the ELP $\Pi_1 \cup \Pi_2 = (\calA_1 \cup
\calA_2, \calE_1 \cup \calE_2, \calR_1 \cup \calR_2)$.

For a set $\calE$ of epistemic literals, a subset $\Phi \subseteq \calE$ of
epistemic literals is called an \emph{epistemic guess} (or, simply, a
\emph{guess}). The following definition provides a way to check whether a set of
interpretations is compatible with a guess~$\Phi$.

\begin{definition}\label{def:compatibility}
  Let $\calA$ be a set of atoms, $\calE$ be a set of epistemic literals over
  $\calA$, and $\Phi \subseteq \calE$ be an epistemic guess. A set $\calI$ of
  interpretations over $\calA$ is called \emph{$\Phi$-compatible w.r.t.\
  $\calE$}, iff 
  \begin{enumerate}
      \item\label{def:compatibility:1} $\calI \neq \emptyset$;
      \item\label{def:compatibility:2} for each epistemic literal $\eneg \ell
	\in \Phi$, there exists an interpretation $I \in \calI$ such that $I
	\not\models \ell$; and
      \item\label{def:compatibility:3} for each epistemic literal $\eneg \ell
	\in \calE \setminus \Phi$, for all interpretations $I \in \calI$ it
	holds that $I \models \ell$.
  \end{enumerate}
\end{definition}

For an ELP $\Pi = (\calA, \calE, \calR)$, the \emph{epistemic reduct}
\cite{ai:ShenE16} of the program $\Pi$ w.r.t.\ a guess $\Phi$, denoted $\Pi^\Phi$,
consists of the rules $\calR^\Phi=\{ r^\neg \mid r \in \calR \}$, where $r^\neg$ is defined
as the rule $r \in \calR$ where all occurrences of epistemic literals $\eneg
\ell \in \Phi$ are replaced by $\top$, and all remaining epistemic negation
symbols $\eneg$ are replaced by default negation $\neg$. Note that, after this
transformation, $\Pi^\Phi=(\calA,\calR^\Phi)$ is a logic program without epistemic
negation\footnote{In fact, $\Pi^\Phi$ may contain triple-negated atoms
$\neg\neg\neg a$. But, according to the definition of the GL-reduct in
\cite{amai:LifschitzTT99}, such formulas are equivalent to simple negated atoms
$\neg a$, and we treat them as such.}. This leads to the following, central
definition.

\begin{definition}\label{def:candidateworldview}
  Let $\Pi = (\calA, \calE, \calR)$ be an ELP. A set $\calM$ of interpretations
  over $\calA$ is a \emph{candidate world view (CWV)} of $\Pi$ if there is an
  epistemic guess $\Phi \subseteq \calE$ such that $\calM =
  \answersets{\Pi^\Phi}$ and $\calM$ is compatible with $\Phi$ w.r.t.\ $\calE$.
  The set of all CWVs of an ELP $\Pi$ is denoted by $\cwvs{\Pi}$. 
\end{definition}

Let us consider an example for illustrative purposes.

\begin{example}\label{ex:cwvs}
  Let $\calA=\{p,p'\}$, $\calE=\{\eneg \neg p\}$, $\Pi=(\calA,\calE,\calR)$ with
  $\calR$ containing only rule $p' \gets \neg \eneg \neg p$, a well-known
  formulation of the closed world assumption proposed in
  \cite{aaai:Gelfond91}\footnote{In fact, in \cite{aaai:Gelfond91}, the author
  proposes the rule $\sneg p \gets \neg \mathbf{M} p$, where $\sneg$ is a third
  kind of negation, usually referred to as strong negation, not considered in
  this paper. It can be simulated by replacing occurrences of $\sneg p$ by a
  fresh atom $p'$ and adding a constraint rule $\gets p, p'$ that excludes $p$
  and $p'$ to hold simultaneously.}.

  We obtain $\cwvs{\Pi}=\{ \{ \{p'\} \} \}$ as guess $\Phi=\emptyset$ is
  compatible with $\answersets{\Pi^{\Phi}=\{p' \gets \neg p\}} = \{\{p'\}\}$,
  while no other guesses are compatible with the answer sets of the respective
  epistemic reducts.
\end{example}

Following the principle of knowledge minimization, a \emph{world view}, in
\cite{ai:ShenE16}, is defined as follows.

\begin{definition}\label{def:worldview}
  Let $\Pi = (\calA, \calE, \calR)$ be an ELP. $\calC\in\cwvs{\Pi}$ is called
  \emph{world view (WV)} of $\Pi$ if its associated guess $\Phi$ is
  subset-maximal, i.e.\ there is no $\calC'\in\cwvs{\Pi}$ with associated guess
  $\Phi'\supset \Phi$.
\end{definition}

Note that in Example~\ref{ex:cwvs} there is only one CWV per program; hence
the associated guesses are subset-maximal, and the sets of CWVs and WVs
coincide.

Note that given two ELPs $\Pi_1 = (\calA_1, \calE_1, \calR_1)$ and $\Pi_2 =
(\calA_2, \calE_2, \calR_2)$, we can always assume that $\calA_1 = \calA_2$ and
$\calE_1 = \calE_2$ without changing the (candidate) world views of the two
programs \cite{aaai:FaberMW19}. In order to simplify our investigation, we will
make use of this assumption when we compare two ELPs.

One of the main reasoning tasks regarding ELPs is the \emph{world view existence
problem}, that is, given an ELP $\Pi$, decide whether a WV (or, equivalently,
CWV) exists. This problem is \SIGMA{3}{P}-complete \cite{ai:ShenE16}.

\section{Uniform Equivalence for ELPs}\label{sec:uniformequivalence}

In this section, we will investigate the uniform equivalence of ELPs, in
particular, focusing on how to extend this concept \cite{iclp:EiterF03} from
logic programs to ELPs. In order to begin this investigation, we will first
define (ordinary) equivalence of two ELPs. The following definition was recently
proposed in \cite{aaai:FaberMW19}.

\begin{definition}\label{def:equivalence}
  Two ELPs are \emph{WV-equivalent} (resp.\ \emph{CWV-equivalent}) iff their
  world views (resp.\ candidate world views) coincide.
\end{definition}

Note that CWV-equivalence immediately implies WV-equivalence.

We now continue by defining uniform equivalence for ELPs. One motivation for
such a kind of equivalence is module optimization: we would like to replace a
module in an ELP, that accepts a set of input facts and provides a set of output
facts, with another (hopefully more efficient) formulation without changing the
semantics (i.e.\ WVs or CWVs). Based on the two equivalence notions defined
above and using ideas from work done in the area of logic programs
\cite{iclp:EiterF03}, we propose two notions of uniform equivalence for ELPs.

\begin{definition}\label{def:uniformequivalence}
  Let $\Pi_1$ and $\Pi_2$ be two ELPs. $\Pi_1$ and $\Pi_2$ are
  \begin{itemize}
    \item \emph{uniformly WV-equivalent} iff, for every set of ground facts $D$,
      $\Pi_1 \cup D$ and $\Pi_2 \cup D$ are WV-equivalent; and
    \item \emph{uniformly CWV-equivalent} iff, for every set of ground facts $D$,
      $\Pi_1 \cup D$ and $\Pi_2 \cup D$ are CWV-equivalent.
  \end{itemize}
\end{definition}

One could be tempted to define uniform equivalence for ELPs simply in terms of
the UE-models \cite{iclp:EiterF03} of the epistemic reducts, for each possible
epistemic guess. However, this approach does not capture our intentions, as the
following example shows:

\begin{example}\label{ex:uemodelsnotenough}
  Take the two ELPs $\Pi_1$ and $\Pi_2$
  \[
  \begin{array}{l@{\qquad}l}
    \Pi_1 = (\calA, \calE, \calR_1) & \Pi_2 = (\calA, \calE, \calR_2)\\
    \calR_1 = \{ p \gets \eneg p \} & \calR_2 = \{ p \gets  \neg p \}
  \end{array}
  \]
  with $\calA = \{p\}$ and $\calE = \{ \eneg p \}$.
  Now, for the guess $\Phi_1 =
  \emptyset$, note that $\Pi_1^{\Phi_1} = \Pi_2^{\Phi_1}$ and hence, trivially,
  the UE-models are also the same. However, for the guess $\Phi_2 = \calE$,
  $\Pi_1^{\Phi_2}$ consists of the rule $p \gets \top$, while $\Pi_2^{\Phi_2}$
  reduces to $p \gets \neg p$. It can be checked that the UE-models of these two
  epistemic reducts w.r.t.\ $\Phi_2$ are not the same and are hence not
  uniformly equivalent in the sense of \cite{iclp:EiterF03}.  However, it turns
  out that the guess $\Phi_2$ can never give rise to a CWV, since it requires
  that there is an answer set not containing $p$, but both $\Pi_1^{\Phi_2}$ and
  $\Pi_2^{\Phi_2}$ require that $p$ is true in all answer sets of the CWV.
\end{example}

The example above implies that, when establishing uniform equivalence for ELPs,
we need a more involved construction. Before we turn to the subject of the
characterization, however, we will first investigate the relationship between
uniform CWV and WV-equivalence.

Clearly, it holds that uniform CWV-equivalence is the stronger notion, as it
directly implies uniform WV-equivalence. It can be shown that this relationship
is strict, and hence the two notions are actually distinct, as the following
proposition states:

\begin{proposition}\label{prop:uniformequivalence}
  Let $\Pi_1$ and $\Pi_2$ be two ELPs. It holds that
  \begin{enumerate}
    \item\label{prop:uniformequivalence:1} when $\Pi_1$ and $\Pi_2$ are
      uniformly CWV-equivalent, then they are uniformly WV-equivalent; and
    \item\label{prop:uniformequivalence:2} the ELPs $\Pi_1$ and $\Pi_2$ may be
      such that $\Pi_1$ and $\Pi_2$ are uniformly WV-equivalent but not
      uniformly CWV-equivalent.
  \end{enumerate}
\end{proposition}

\begin{proof}
  (\ref{prop:uniformequivalence:1}) As observed after
  Definition~\ref{def:equivalence}, if two ELPs are CWV-equivalent then they are
  WV-equivalent. This holds, in particular, for any set of facts $D$, and the
  ELPs $\Pi_1 \cup D$ and $\Pi_2 \cup D$.

  \smallskip\noindent (\ref{prop:uniformequivalence:2}) We will prove this by
  example. Take the ELPs $\Pi_1 = (\calA, \calE, \calR_1)$ and $\Pi_2 = (\calA,
  \calE, \calR_2)$ built as follows. Let $\calR$ be the following set of rules:
 \begin{align*}
  \calR = \{ a \vee b & \gets  \eneg \neg a, \eneg \neg b;\\
    a & \gets  b;\\
    b &\gets  a;\\
    c &\gets  d, \neg a, \neg b;\\
    d &\gets c, \neg a, \neg b \}
  \end{align*}
  Now, let
  \begin{align*}
    \calR_1 &= \calR \cup \{ c \vee d \gets \neg a, \neg b \}\\
    \calR_2 &= \calR \cup \{ c \gets \neg d, \neg a, \neg b;\, d \gets \neg c,
    \neg a, \neg b \}.
  \end{align*}
  We will show that $\Pi_1$ and $\Pi_2$ are uniformly WV-equivalent, but not
  uniformly CWV-equivalent.
  
  Note that $\calA = \{ a, b, c, d \}$ and $\calE = \{ \eneg \neg a, \eneg \neg
  b \}$. To prove our claim, let us first examine the first three rules of
  $\calR$. From these rules, it is not difficult to check that there are two
  epistemic guesses that lead to CWVs, namely $\Phi_1 = \emptyset$ and $\Phi_2 =
  \calE$. The CWV for $\Phi_2$ is the set $\{ \{ a, b \} \}$. Note that $\Phi_2$
  is subset-maximal, and hence this set is also a WV. Note further that adding
  any set of facts $D \subseteq \calA$ to $\Pi_1$ or $\Pi_2$ will simply change
  the WV to $\{ \{ a, b \} \cup D \}$, which is still a valid WV w.r.t.\ guess
  $\Phi_2$ for both ELPs.  However, the CWVs w.r.t.\ guess $\Phi_1$ differ
  already for $D = \emptyset$: for $\Pi_1 \cup D$ it is $\{ \{ a, b, c, d \}
  \}$, whereas for $\Pi_2 \cup D$ no CWV exists. Hence, we have that $\Pi_1$ and
  $\Pi_2$ are uniformly WV-equivalent, but not uniformly CWV-equivalent, as
  desired.
\end{proof}

The above result shows an interesting distinction between uniform equivalence
and strong equivalence when regarding ELPs. As shown in \cite{aaai:FaberMW19},
the different notions of strong equivalence considered therein coincide (that
is, regarding strong equivalence w.r.t.\ WVs or CWVs does not make a
difference), this is not the case for uniform equivalence, where there is an
actual distinction between uniform CWV- and uniform WV-equivalence.

A further observation that can be made is that both forms of uniform equivalence
for ELPs strictly generalize the notion of uniform equivalence for ASP, as the
following result shows.

\begin{theorem}\label{thm:generalization}
  Let $\Pi_1 = (\calA, \calR_1)$ and $\Pi_2 = (\calA, \calR_2)$ be two logic
  programs, and $\Pi_1' = (\calA, \calE, \calR_1)$ and $\Pi_2' = (\calA, \calE,
  \calR_2)$ be two ELPs containing the same rules, respectively, and where
  $\calE = \emptyset$. Then, the following three statements are equivalent:
  \begin{enumerate}
    \item\label{thm:generalization:1} $\Pi_1$ and $\Pi_2$ are uniformly
      equivalent,
    \item\label{thm:generalization:2} $\Pi_1'$ and $\Pi_2'$ are uniformly
      CWV-equivalent, and
    \item\label{thm:generalization:3} $\Pi_1'$ and $\Pi_2'$ are uniformly
      WV-equivalent.
  \end{enumerate}
\end{theorem}

\begin{proof}
  Note that, since $\calE = \emptyset$, both ELPs $\Pi_1'$ and $\Pi_2'$ have at
  most one CWV (and hence WV), namely the set $\answersets{\Pi_1}$ and
  $\answersets{\Pi_2}$, respectively, in case these sets are non-empty.
  Otherwise, if $\Pi_1$ or $\Pi_2$ is inconsistent, then $\Pi_1'$ or $\Pi_2'$ do
  not have any CWVs, respectively.
  
  \paragraph{(\ref{thm:generalization:2}) $\Leftrightarrow$
  (\ref{thm:generalization:3}).} Since $\Pi_1'$ and $\Pi_2'$ have at most one
  CWV that corresponds to the guess $\Phi = \calE = \emptyset$, the notions of
  uniform WV-equivalence and uniform CWV-equivalence coincide.

  \paragraph{(\ref{thm:generalization:1}) $\Rightarrow$
  (\ref{thm:generalization:2}).} By assumption, $\Pi_1$ and $\Pi_2$ are
  uniformly equivalent. Towards a contradiction, assume that there is a set of
  facts $D \subseteq \calA$, such that $\Pi_1' \cup D$ is not CWV-equivalent to
  $\Pi_2' \cup D$. Note that, since the only epistemic guess possible is the
  guess $\Phi = \calE = \emptyset$, any non-empty set of answer sets satisfies
  $\Phi$. Further, note that for $i \in \{ 1, 2 \}$ it holds that $(\Pi_i' \cup
  D)^\Phi = \Pi_i'^\Phi \cup D = \Pi_i \cup D$. Hence, we have that
  $\answersets{\Pi_1 \cup D} \neq \answersets{\Pi_2 \cup D}$, contradicting our
  assumption.

  \paragraph{(\ref{thm:generalization:2}) $\Rightarrow$
  (\ref{thm:generalization:1}).} This follows from a similar argument as the one
  above.
\end{proof}

Having defined the notions of uniform equivalence for ELPs, we aim to
characterize it in a similar fashion as was done for strong equivalence for ELPs
in \cite{aaai:FaberMW19}, and for logic programs in \cite{tplp:Turner03}.
Unfortunately, it seems that an ``interesting'' characterization, as in these
two papers, is not possible for uniform equivalence of ELPs. Due to the complex
interactions between epistemic guesses, answer sets, and the sets of facts
added, only a very straightforward characterization of uniform equivalence for
ELPs is possible. We formulate this as a so-called UE-function, in the spirit of
the SE-function for strong equivalence as given in \cite{aaai:FaberMW19}.

\begin{definition}\label{def:uefunction}
  Let $\Pi = (\calA, \calE, \calR)$ be an ELP, and let $W \in \{ \text{CWV},
  \text{WV} \}$. Then, $\calU\calE^W_\Pi: 2^\calE \times 2^\calA \to 2^\calA$ is
  called the $W$-UE-function of $\Pi$ iff, for any epistemic guess $\Phi
  \subseteq \calE$ and any set of facts $D \subseteq \calA$, it holds that \[
    \calU\calE^W_\Pi(\Phi, D) = \left\{
    \begin{array}{@{}lr@{}}
      \calM &\qquad \text{if } \calM \text{ is a CWV of type } W \text{ of } \Pi
      \text{ w.r.t.\ } \Phi\\
      \emptyset &\qquad \text{otherwise,}
    \end{array}
    \right.
  \] where $\calM = \answersets{(\Pi \cup D)^\Phi}$.
\end{definition}

As can be seen, the characterization is rather straightforward (for a given ELP
$\Pi$, it maps an epistemic guess and a set of facts to the CWV or WV that
arises w.r.t.\ the guess when adding the set of facts to $\Pi$). The following
result follows immediately from the construction of the UE-function:

\begin{theorem}\label{thm:uefunction}
  For $W \in \{ \text{CWV}, \text{WV} \}$, two ELPs $\Pi_1$ and $\Pi_2$ are
  $W$-equivalent iff their $W$-UE-functions coincide.
\end{theorem}

While this characterization thus is far less interesting than the
characterization for strong equivalence, it seems that the multiple layers
involved in computing world views of ELPs make a more interesting construction, which
tries to directly use UE-models from classical logic programming
\cite{iclp:EiterF03}, impossible. Further evidence of this will be presented in
the next section, where we investigate the computational complexity of deciding
uniform equivalence. While deciding strong equivalence for ELPs is
\co\NP-complete, we will see that the same task for uniform equivalence is much
harder, making it unlikely that an elegant compact representation, like the
UE-models proposed in \cite{iclp:EiterF03}, or the SE-function from
\cite{aaai:FaberMW19} can be found.

Before turning our attention to this topic, however, let us briefly return to
our example from Section~\ref{sec:introduction}. Recall that in
Example~\ref{ex:running1}, we have seen two versions to formulate the CWA using
ELPs. We shall investigate these two formulations w.r.t.\ their uniform
equivalence.

\begin{example}\label{ex:running2}
  It turns out that the two formulations of CWA shown in
  Example~\ref{ex:running1} are in fact both uniformly CWV-equivalent and
  uniformly WV-equivalent ELPs. This can be verified by constructing the
  relevant UE-functions. Let $\Pi_{\text{Gelfond}} = (\calA, \calE,
  \calR_{\text{Gelfond}})$ be the formulation of the CWA from
  \cite{aaai:Gelfond91}, and $\Pi_{\text{ShenEiter}} = (\calA, \calE,
  \calR_{\text{ShenEiter}})$ be the formulation from \cite{ai:ShenE16}, where
  $\calA = \{ p', p \}$ and $\calE = \{ \eneg p, \eneg \neg p \}$. Omitting all
  combinations of epistemic guesses and sets of facts where the UE-functions
  return $\emptyset$, the UE-functions, for both $W \in \{ \text{WV}, \text{CWV}
  \}$, look as follows:
  \begin{align*}
    \calU\calE^W_{\Pi_{\text{Gelfond}}} & =
    \calU\calE^W_{\Pi_{\text{ShenEiter}}} = \{\\
    & ( \{ \eneg p \}, \emptyset, \{ \{ p' \} \} ),\\
    & ( \{ \eneg p \}, \{ p' \}, \{ \{ p' \} \} ),\\
    & ( \{ \eneg \neg p \}, \{ p \}, \{ \{ p', p \} \} ),\\
    & ( \{ \eneg \neg p \}, \{ p', p \}, \{ \{ p', p \} \} )\\
    \}
  \end{align*}
  From this, since the UE-functions for the two formulations coincide, we
  observe that in the context of uniform equivalence, these two formulations
  are, for all intents and purposes, interchangeble.
\end{example}

With the above example, we are able to formally establish the following result,
comparing our notions of uniform equivalence for ELPs to established notions of
equivalence. We observe that, as expected, uniform equivalence is strictly
stronger than (ordinary) equivalence, but strictly weaker than strong
equivalence.

\begin{theorem}\label{prop:comparison}
  For ELPs, it holds that
  \begin{enumerate}
    \item\label{prop:comparison:1} strong equivalence is strictly stronger than
      uniform CWV-equivalence;
    \item\label{prop:comparison:2} uniform CWV-equivalence is strictly stronger
      than uniform WV-equivalence; 
    \item\label{prop:comparison:3} uniform CWV-equivalence is strictly stronger
      than (ordinary) CWV-equivalence;
    \item\label{prop:comparison:4} uniform WV-equivalence is strictly stronger
      than (ordinary) WV-equivalence.
  \end{enumerate}
\end{theorem}

\begin{proof}
  Several observations follow trivially from the relevant definitions: strong
  equivalence implies uniform CWV-equivalence (since sets of atoms are also
  ELPs), which clearly implies uniform WV-equivalence (since WVs are a subset of
  CWVs). Finally, since uniform CWV-equivalence (resp.\ WV-equivalence) require
  that the CWVs (resp.\ WVs) are the same for any set of added atoms---in
  particular, the empty set of atoms---they directly imply (ordinary)
  CWV-equivalence (resp.\ WV-equivalence).

  To establish that the separations between these equivalence notions are indeed
  strict, we make use of several separating examples. Statement
  (\ref{prop:comparison:1}) is shown by Examples~\ref{ex:running1} and
  \ref{ex:running2}, which exhibit two ELPs that are uniformly CWV-equivalent,
  but not strongly equivalent. Statement (\ref{prop:comparison:2}) follows from
  statement (\ref{prop:uniformequivalence:2}) of
  Proposition~\ref{prop:uniformequivalence}. Finally, statements
  (\ref{prop:comparison:3}) and (\ref{prop:comparison:4}) follow from the fact
  that uniform CWV-equivalence and uniform WV-equivalence both generalize
  uniform equivalence for ground logic programs (cf.\
  Theorem~\ref{thm:generalization}), and, in this case, uniform equivalence is
  already strictly stronger than ordinary equivalence; see e.g.\
  \cite[Example~10]{iclp:EiterF03}.
\end{proof}

\section{Complexity of ELP Uniform Equivalence}\label{sec:complexity}

Having defined our characterization of uniform equivalence for ELPs, in this
section, we will now focus on the question of the computational complexity of
deciding whether two ELPs are uniformly equivalent. It turns out that this task
is of similar hardness as deciding the CWV existence problem for ELPs, that is,
on the third level of the polynomial hierarchy. Hence, it is one level higher in
the polynomial hierarchy than for plain ground (disjunctive) logic programs
under the stable model semantics, for which uniform-equivalence checking is
\PI{2}{P}-complete \cite[Theorem~10]{iclp:EiterF03}. The following result states
this formally:

\begin{theorem}\label{thm:complexity}
  Deciding uniform CWV-equivalence of two ELPs is \PI{3}{P}-complete.
\end{theorem}

\begin{proof}
  For this proof, assume that $\Pi_1 = (\calA, \calE, \calR_1)$ and $\Pi_2 =
  (\calA, \calE, \calR_2)$ are two ELPs (w.l.o.g.\ over the same set of atoms
  and epistemic literals).

  \paragraph{Upper Bound.} As stated in \cite[Theorem~4]{ai:ShenE16}, given an
  epistemic guess $\Phi$ and an ELP $\Pi$, verifying that $\Pi$ has a CWV
  w.r.t.\ $\Phi$ can be done in \DP{2}{P}, and hence via two calls to a
  \SIGMA{2}{P}\ oracle. We therefore obtain a straightforward guess-and-check
  algorithm that runs in non-deterministic polynomial time with a \SIGMA{2}{P}\
  oracle, checks non-uniform equivalence between two ELPs, and works as follows:
  guess a set of atoms $D \subseteq \calA$, an epistemic guess $\Phi \subseteq
  \calE$, and a set of facts $M \subseteq \calA$. Then, use a
  \SIGMA{2}{P}-oracle to check that one of the following two conditions hold:
  (i) $\Phi$ leads to a CWV for $\Pi_1 \cup D$, but not for $\Pi_2 \cup D$, or
  (ii) $\Phi$ leads to a CWV for both $\Pi_1 \cup D$ and $\Pi_2 \cup D$, but
  that $M$ is an answer set that exists only in exactly one of these two CWVs.

  \paragraph{Lower Bound.} We will show \PI{3}{P}-hardness via reduction from
  3-QBF solving. We will construct two ELPs $\Pi_1$ and $\Pi_2$ such that they
  are uniformly equivalent iff a given 3-QBF is unsatisfiable. To this end, we
  will make use of the reduction from 3-QBF solving to CWV existence offered in
  \cite[Proof of Theorem~5]{ai:ShenE16}, on which our reduction is based.  Let
  $\exists \varsX \forall \varsY \exists \varsZ \; \Psi(\varsX, \varsY, \varsZ)$
  be a 3-QBF formula in conjunctive normal form, where each clause has the form
  $\ell_1 \vee \ell_2 \vee \ell_3$, where each $\ell$ is a literal over the
  variables in $\varsX \cup \varsY \cup \varsZ$.  In \cite{ai:ShenE16}, it is
  assumed w.l.o.g.\ that the 3-QBF evaluates to true whenever all variables in
  $\varsY$ are replaced by $\top$. This does not change the hardness of the
  problem, and we make use of the same assumption. For a 3-QBF formula as above,
  we construct the ELP $\Pi_1 = (\calA, \calE, \calR_1)$ over the atoms $\calA =
  \{ w, \overline{w} \mid w \in \varsX \cup \varsY \cup \varsZ \} \cup \{
  \relation{false}, \relation{sat} \}$ using the well-known technique of
  saturation \cite{amai:EiterG95}. $\Pi_1$ contains the following set of rules,
  where $\ell^*$ converts a literal $a$ into atom $\overline{a}$ and literal
  $\neg a$ into atom $a$:
  %
    %
    \begin{itemize}
      \item for each $x \in \varsX$:
	\begin{align}
	  x & \gets \eneg \overline{x}, \label{qbf:1}\\
	  \overline{x} & \gets \eneg x; \label{qbf:2}
	\end{align}
      \item for each $y \in \varsY$:
	\begin{align}
	  y & \gets \neg \overline{y}, \label{qbf:3}\\
	  \overline{y} & \gets \neg y, \label{qbf:4}\\
	  \bot & \gets \neg \eneg y, \label{qbf:5}\\
	  \bot & \gets \neg \eneg \overline{y}; \label{qbf:6}
	\end{align}
      \item for each $z \in \varsZ$:
	\begin{align}
	  z \vee \overline{z} & , \label{qbf:7}\\
	  z & \gets \relation{sat}, \label{qbf:8}\\
	  \overline{z} & \gets \relation{sat}; \label{qbf:9}
	\end{align}
      \item for each clause $\ell_1 \vee \ell_2 \vee \ell_3$ in $\Psi$: 
	\begin{align}
	  \relation{sat} & \gets \ell_1^*, \ell_2^*, \ell_3^*; \label{qbf:10}
	\end{align}
      \item and the two rules
	\begin{align}
	  \relation{false} & \gets \eneg \relation{false}, \eneg \neg
	  \relation{sat}, \label{qbf:11}\\
	  \bot & \gets \neg \eneg \relation{false}. \label{qbf:12}
	\end{align}
    \end{itemize}
    %
    %
    
  The construction of $\Pi_2 = (\calA, \calE, \calR_2)$ differs from $\Pi_1$ in
  only one respect: for $\Pi_2$, $\ell^*$ converts literals into $\top$.  Note
  that, therefore, $\Pi_2$ contains the fact $\relation{sat}$. This completes
  the main part of our construction. Let us now explore how our construction
  works. From \cite{ai:ShenE16}, we have that program $\Pi_1$ has a CWV iff the
  3-QBF $\exists \varsX \forall \varsY \exists \varsZ \Psi$ is satisfiable, and
  hence, conversely, $\Pi_1$ has no CWVs iff the QBF is unsatisfiable (since, in
  this case, in the GL-reduct, the atom $\relation{sat}$ is always derived, and
  hence, constraint~(\ref{qbf:11}) destroys any potential CWV). Note that any
  CWV $\calM$ for $\Pi_1$ has the following structure: a guess $\Phi$ leading to
  $\calM$ will contain a subset of $\{ \eneg x, \eneg \overline{x} \mid x \in
  \varsX \}$ representing an assignment on the variables of $\varsX$. Further,
  $\Phi \supseteq \{ \eneg y, \eneg \overline{y} \mid y \in \varsY \}$, since
  each answer set in the CWV represents precisely one assignment on the
  variables $\varsY$, and all possible such assignments must appear in the CWV.
  Finally, $\eneg \relation{false} \in \Phi$ and $\eneg \neg \relation{sat}
  \not\in \Phi$, via constraint~(\ref{qbf:11}). For the precise reasoning behind
  this construction, please see \cite[Proof of Theorem~5]{ai:ShenE16}.
  
  Towards our goal, we must show two things: (a) in cases where the 3-QBF
  $\exists \varsX \forall \varsY \exists \varsZ \Psi$ is satisfiable, there
  exists a set of atoms $D \subseteq \calA$, such that $\Pi_1 \cup D$ and $\Pi_2
  \cup D$ are not equivalent (i.e.\ have differing CWVs), and hence, $\Pi_1$ and
  $\Pi_2$ are not uniformly equivalent; and (b) in cases where the 3-QBF is
  unsatisfiable for all sets of facts $D \subseteq \calA$ it holds that $\Pi_1
  \cup D$ is equivalent to $\Pi_2 \cup D$ (i.e.\ they have the same CWVs), and
  hence, $\Pi_1$ and $\Pi_2$ are uniformly equivalent.

  Showing (a) is straightforward: simply take $D = \emptyset$. $\Pi_1$ has a CWV
  (via correctness of the reduction in \cite{ai:ShenE16} as explained above),
  whereas $\Pi_2$ does not have a CWV, since the atom $\relation{sat}$ is always
  derived in any GL-reduct of $\Pi_2$, destroying each potential CWV.  Showing
  (b) is a little more involved. We will show this by contradiction. To this
  end, assume that the 3-QBF is unsatisfiable, but some set of facts $D
  \subseteq \calA$ exists, such that $\Pi_1 \cup D$ and $\Pi_2 \cup D$ are not
  equivalent. Our plan is to show that $D$ cannot contain any atoms from
  $\calA$, but also cannot be empty. Let us look at the atoms in $\calA$ in
  turn.
  \begin{description}
    \item[$\relation{sat} \in D$:] in this case, the only difference between
      $\Pi_1$ and $\Pi_2$, namely rules of the form~(\ref{qbf:10}), disappears,
      and hence, $\Pi_1 \cup D$ and $\Pi_2 \cup D$ cannot have differing CWVs; a
      contradiction.
    \item[$\relation{false} \in D$:] in this case, the atom $\relation{false}$
      is true in every answer set, regardless of the epistemic guess $\Phi$, in
      both $\Pi_1 \cup D$ and $\Pi_2 \cup D$. Hence, $\eneg \relation{false}
      \not\in \Phi$. But then, rule~(\ref{qbf:12}) becomes $\bot \gets \neg \neg
      \relation{false}$ in the epistemic reduct w.r.t.\ $\Phi$, and no answer
      set can both contain $D$ and satisfy this constraint; a contradiction.
    \item[$\{ y, \overline{y} \mid y \in \varsY \} \cap D \neq \emptyset$:] this
      case is similar to the case of $\relation{false}$. If any such atom $y$ or
      $\overline{y}$ is in $D$, and hence true in every answer set of any
      epistemic reduct, then constraints~(\ref{qbf:5}) and~(\ref{qbf:6}) will
      prevent that epistemic reduct from having any answer sets for both $\Pi_1$
      and $\Pi_2$; a contradiction.
    \item[$D \subseteq \{ w, \overline{w} \mid w \in \varsX \cup \varsZ \}$:]
      since the 3-QBF is unsatisfiable, we know that for any assignment on the
      variables in $\varsX$ and $\varsZ$ there is an assignment on the variables
      in $\varsY$ such that $\Psi$ is false. Since, from the last paragraph, we
      know that $D \cap \{ y, \overline{y} \mid y \in \varsY \} = \emptyset$, we
      have that whatever assignment on the variables $\varsX$ and $\varsZ$ is
      fixed via the atoms in $D$ (in particular, also when $D = \emptyset$),
      there will always be an assignment on the variables $\varsY$, represented
      by the atoms $\{ y, \overline{y} \mid y \in \varsY \}$, such that the atom
      $\relation{sat}$ will be derived in the GL reduct of $\Pi_1 \cup D$,
      irrespective of the guess $\Phi$, and hence the assignment on the
      variables $\varsX$. Hence, again, $\Pi_1 \cup D$ and $\Pi_2 \cup D$ have
      no CWVs; a contradiction.
  \end{description}
  We thus have that $D$ cannot be empty, but also cannot contain any atoms from
  $\calA$, and hence, cannot exist. Since, by construction, all CWVs of $\Pi_1$
  and $\Pi_2$ are also WVs (as the respective epistemic guesses are never in a
  subset-relationship), the above holds for both uniform CWV- and uniform
  WV-equivalence. This concludes the proof.
\end{proof}

From the proof of the above theorem, we immediately obtain the following
statement for uniform WV-equivalence, which follows from the fact that our
lower-bound construction employs an encoding for 3-QBF where the set of WVs and
CWVs always coincide, and hence, the two ELPs $\Pi_1$ and $\Pi_2$ in this
construction are uniformly WV-equivalent iff they are uniformly CWV-equivalent.

\begin{theorem}
  Deciding uniform WV-equivalence for two ELPs is \PI{3}{P}-hard.
\end{theorem}

Note, however, that our upper bound construction does not give an upper bound
for uniform WV-equivalence, since verifying that some CWV is a WV is
\PI{3}{P}-hard itself \cite{ai:ShenE16}.

\section{Conclusions}\label{sec:conclusions}

In this paper, we have defined and studied the notion of uniform equivalence for
epistemic logic programs. Programs are uniformly WV- or CWV-equivalent if they
yield the same world views or candidate world views, respectively. In contrast
to strong equivalence for ELPs, the two notions (for WV and CWV) do not
coincide, but they generalize uniform equivalence for standard logic programs
interpreted using the Answer Set semantics. We also provided a characterization
of both notions of uniform equivalence by means of a UE-function, in the spirit
of the SE-function of \cite{aaai:FaberMW19}. However, unlike the SE-function
this characterization is relatively straightforward and provides only little
insight into the problem. While this reduces the potential impact of the
characterization, it appears that one cannot do better. In fact, we show that
deciding uniform equivalence on ELPs is at least \PI{3}{P}-hard and thus
probably much harder than deciding strong equivalence on ELPs, which is
\co\NP-complete. This result provides a further indication that a more compact
representation of the UE-function is unlikely to exist.

For future work, it would be interesting to see whether other forms of
equivalences between ELPs exist that are less restrictive than strong
equivalence but more restrictive than uniform equivalence, and, ideally, for
which the decision problem also lies between the respective complexities of
deciding uniform and strong equivalence.

\paragraph*{Acknowledgements} Wolfgang Faber and Michael Morak were supported by
the S\&T Cooperation CZ 05/2019 ``Identifying Undoable Actions and Events in
Automated Planning by Means of Answer Set Programming.'' Stefan Woltran was
supported by the Austrian Science Fund (FWF) under grant number Y698.

\bibliographystyle{acmtrans}
\bibliography{references}

\end{document}